\documentclass[11pt]{article}

\newcommand{\ifabs}[2]{#2}
\usepackage{amsmath,amsfonts,amsthm,amssymb}
\usepackage{algorithm}
\usepackage{graphics,graphicx,color,url}
\usepackage{epsfig}
\usepackage{epsf}
\usepackage{fullpage}
\usepackage{latexsym}
\usepackage{stmaryrd}
\usepackage{verbatim}

\newtheorem{theorem}{Theorem}[section]
\newtheorem{lemma}{Lemma}[section]
\newtheorem{corollary}{Corollary}[section]
\newtheorem{claim}{Claim}
\newtheorem{definition}{Definition}[section]
\newtheorem{remark}{Remark}[section]

\newcommand{\mainmechanism}{{\sc Promoting-Testing-Selling Mechanism}}
\newcommand{\costsharing}{{\sc Cost Sharing Mechanism}}
\newcommand{\calA}{{\cal A}}

\newcommand{\calD}{{\cal D}}

\newcommand{\calF}{{\cal F}}

\newcommand{\calM}{{\cal M}}

\newcommand{\calP}{{\cal P}}

\newcommand{\vecb}{{\mathbf b}}
\newcommand{\vect}{{\mathbf t}}

\begin{document}
%\title{Competitive Auctions for Markets with Externalities}
\title{Competitive Auctions for Markets with Positive Externalities}

%\numberofauthors{2}
\author{\\Nick Gravin\thanks{Division of Mathematical Sciences, School of Physical and Mathematical Sciences, Nanyang Technological University, Singapore. Email: {\tt ngravin@pmail.ntu.edu.sg}.} \thanks{St. Petersburg Department of Steklov Mathematical Institute RAS, Russia. Email: {\tt gravin@pdmi.ras.ru}. Work done while visiting Microsoft Research Asia.}
\and \\Pinyan Lu\thanks{Microsoft Research Asia. Email: {\tt
pinyanl@microsoft.com}.}}

\date{}

\maketitle

\begin{abstract}
In digital goods auctions, there is an auctioneer who sells an item
with unlimited supply to a set of potential buyers, and the objective
is to design truthful auction to maximize the total profit of the
auctioneer. Motivated from an observation that the values of buyers
for the item could be interconnected through social networks, we
study digital goods auctions with positive externalities among the
buyers. This defines a multi-parameter auction design problem where
the private valuation of every buyer is a function of other winning
buyers. The main contribution of this paper is a truthful
competitive mechanism for subadditive valuations. Our competitive
result is with respect to a new solution benchmark
$\mathcal{F}^{(3)}$; on the other hand, we show a surprising impossibility result
if comparing to the benchmark $\mathcal{F}^{(2)}$, where the latter
has been used quite successfully in digital goods auctions without
extenalities~\cite{Goldberg2006}. Our results from
$\mathcal{F}^{(2)}$ to $\mathcal{F}^{(3)}$ could be considered as
the loss of optimal profit at the cost of externalities.
\end{abstract}

%\begin{abstract}
%We consider a problem of designing auction among buyers with positive externalities.
%Our mechanism design model is a generalization of the classical auctions for selling
%a good of unlimited supply with multi-parameter bidders, each reporting its valuations
%over possible sets of winners. Our main goal is to maximize auctioneer's revenue.
%Recently similar agenda was proposed in \cite{Mirrokni},
%where authors consider single parameter domain and submodular valuations of bidders.
%In contrast to this work we deal with the multi-parameter case and subadditive
%valuations, though using another classical approach~\cite{Goldberg2006},
%namely we aim at designing a competitive mechanism with respect to the best uniform price benchmark.
%The main contribution of this paper is the first competitive universally truthful mechanism with respcet
%to this benchmark. Our mechanism works in polynomial of the number of agents time, that is
%it may ask only a polynomial number of queries to each agent about its valuations for certain
%sets.
%\end{abstract}

%\begin{abstract}
%We consider a problem of designing auction among buyers with positive externalities.
%In this scenario we need to organize an auction for selling a good with unlimited supply
%assuming that bidders may increase valuations of the rest buyers when getting the good.
%For this natural multi-parameter mechanism design problem we give a competitive
%universally truthful mechanism w.r.t. the best uniform price benchmark.
%\end{abstract}

\ifabs{\thispagestyle{empty}
\newpage
\setcounter{page}{1}}

\section{Introduction}
%Economy+externalities
In economics, the term externality is used to describe those
situations where the private costs or benefits to the producers or
purchasers of a good or service differs from the total social costs
or benefits entailed in its production and consumption. In this
context a benefit is called positive externality, while a cost is
called negative. One need not go too far to find examples of positive
external influence in the digital and communications markets, when a
customer's decision to buy a good or purchase a service strongly
relies on its popularity among his friends or generally among other
customers, e.g. instant messenger and cell phone users will want a
product that allows them to talk easily and cheaply with their
friends. Another good example may be given by social networks, where
a user appreciates higher a membership in the network if many of his
friends are already using it. There exist a number
of applications like quite popular Farm Ville in online social network
Facebook, where a user would have more fun when playing it with his
friends. In fact, quite a few of such applications explicitly
reward players with a big number of friends.

On the other hand, the negative external effects take place when a
potential buyer, e.g. a big company, incures a great loss, if a subject
it fights for, like small firm or company, comes to its direct
competitor. Another well studied example related to computer science
may be given by allocation of advertisment slots \cite{Aggarwal,Ghosh08,Ghosh10,Giotis,Gomes,Kempe},
where every customer would like to see a smaller number of competitors'
advertisements on a web page that contains his own advert.
One may also face mixed externalities as in the case of salling nuclear
weapons \cite{Jehiel96R}, where countries would like to see their allies
win the auction rather than their foes.

%Mechanism design agenda.
In contrast, we investigate the problem of {\em mechanism design} for the
auctions with positive externalities. We study the scenario where
an auctioneer sells the goods, of no more than one item each, into the
hands of customers. We define a model for externalities among
buyers in the sealed-bid auction with unlimited supply of the good.
Those kind of auctions arise naturally in the digital
markets, where making a copy of a good (e.g. cd with songs and games or
extra copy of online application) has a negligible cost compared to
the final price and can be done at any time the seller chooses.

%Mentioning of EC paper on positive externalities
Recently similar agenda has been introduced in the paper \cite{Mirrokni}, where
athours consider bayesian framework and study positive externalities in social networks
with single-parameter bidders and submodular valuations.
The model in the most general form can be described by a number of bidders $n$,
each with a non-negative private valuation function $v_i(S)$ depending on the possible
winning set $S$. This is natural mechanism design multi-parameter model
that may be considered as a generalization of classical auctions with unlimited supply, i.e.
auctions where the amount of items being sold is greater than the number
of buyers.

%Benchmark
Traditionally the main question arizing in these kind of situations is to maximize
seller's revenue. In the literature on classical aucitons without any
externalities there were developed diverse approaches to this question.
In the current work we pick a different from byesian framework classical
benchmark (cf. \cite{Goldberg2006}), namely the best-uniform-price benchmark called $\calF$.
There one seeks to maximize the ratio of the mechanism's revenue to the revenue of $\calF$
taken in the worst case over all possible bids. In particular a mechnaism is
called competitive if such ratio is bounded by some uniform constant for each possible bid.
However, it was shown that there is no competitive truthful mechanism
w.r.t. $\calF$, and therefore to get round this problem, there was
proposed a slightly modified benchmark $\calF^{(2)}$. The only difference of
$\calF^{(2)}$ to $\calF$ is in one additional requirement
that at least two buyers should be in a winning set. Thus $\calF^{(2)}$
becomes a standard benchmark in analyzing digital auctions.
Similarly to $\calF^{(2)}$ one may define benchmark $\calF^{(k)}$
for any fixed constant $k$. It turns out that the same benchmarks
can be naturally adopted to the case of positive externalities.
Surprisingly $\calF^{(2)}$ fails to serve as a benchmark in social
networks with positive externalities, i.e. no competitive mechanism exists
w.r.t. $\calF^{(2)}$. Therefore, we go further and consider the next natural
candidate for the benchmark, that is $\calF^{(3)}$.

%Our contribution
The main contribution of the current paper is a universally truthful
competitive mechanism for the general multi-parameter model with
subadditive valuations (substantially broader class than submodular)
w.r.t. $\calF^{(3)}$ benchmark. As a complement we furnish this result with the proof that
no truthful mechanism can archieve constant ratio w.r.t. $\calF^{(2)}$. In order to do
so we introduce a restricted model with single private parameter which in some respects
resamble that considered in \cite{Mirrokni}; further for this restricted model
we adduce a simple geometric characterization of all truthful mechanisms and based on this characterization
then show that there exists no competitive truthful mechanism.

%Discussion section
To be completely consistent we admit that besides claimed
monotonicity (positive externalities) and subadditivity restrictions on the
valuation functions we additionally require that each agent derives zero
value when not obtaining the good. First, this is reallistic
assumption, e.g. without a messanger or online application any customer derives zero utility regardless of
how many his friends got it. Second, in discussion Section we argue that the
later is indeed necessary condition in order to get a competitive mechanism.
We also consider some other natural extentions and show that all of them
fails to archieve a constant ratio w.r.t. any benchmark $\calF^{(k)}$ for a
fixed $k$.

\subsection{Related Works}
%Papers on pricing for externalities
Many studies on externalities in the direction of pricing and
marketing strategies over social networks have been conducted over
the past years. They have been caused in many ways by the
development of social-networks on the Internet, which has allowed
companies to collect information about users and their relationships.

%Influence maximization.
The earlier works were generally devoted to the influence
maximization problems (see Chapter 24 of \cite{Kleinberg07}). For
instance, Kempe {\it et.al.} \cite{Kempe2005} study the
algorithmic question of searching a set of nodes in a social network
of highest influence. From the economics literature one could name
such papers as \cite{Pekka07}, which studies the effect of network
topology on a monopolist's profits and \cite{Domingos01}, which studies a
multi-round pricing game, where a seller may lower his price in an
attempt to attract low value buyers. As usual for economics
literature all of these works take no heed of algorithmic
motivation.

%Revenue maximization.
More recently there emerged several papers~\cite{Akhlaghpour2009,Arthur2009,Hartline2008}
studying the question of revenue maximization and work studing the post price mechanisms
~\cite{Akhlaghpour,Anari,Candogan,Hartline2008}.

%revenue maximization wo. any extern. (related papers, overview of results)
We could not go by without a mention of a beautiful line of research
on revenue maximization for classical auctions, where the
objective is to maximize the seller's revenue compared to a benchmark
in the worst case. We cite here only some papers that are
most relevant to our setting
~\cite{AlaeiMS09,FeigeFHK05,FiatGHK02,Goldberg2006,
Hartline2005}. With respect to the
refined best-uniform-price benchmark $\calF^{(2)}$ a number of
mechanisms with constant competitive ratio were obtained; each
subsequent paper improving the competitive ratio of the previous one
~\cite{FeigeFHK05,FiatGHK02,Goldberg2006,Hartline2005}. The
best known current mechanism by Hartline and
McGrew~\cite{Hartline2005} has a ratio of $3.25$. On the other
hand a lower bound of $2.42$ has been proved in
\cite{Goldberg2006} by Goldberg {\it et.al.}. The question of
closing the gap still remains open.

\subsection*{Organization of the Paper}
We begin with all necessary definitions in Section \ref{sec:definitions}.
Section \ref{section:Mechanism} presents a
competitive mechanism w.r.t. to benchmark $\calF^{(3)}$
for the general model with multi parameter bidding.
In Section \ref{sec:char} we give a geometric characterization of truthful
mechanism for some restricted single-parameter cases, which we need further
is Subsection \ref{subsection:F2} in order to show the impossibility of designing
a competitive mechanism w.r.t. $\calF^{(2)}$. Section \ref{sec:char} is also
furnished with a simpler and better competitive mechanism in Subsection \ref{subsec:additive}
for one of these special cases w.r.t. a stronger $\calF^{(2)}$ benchmark.
We conclude with the Section \ref{section:Extensions} where we discuss possible
extensions of the model and give a list of open questions.

\section{Preliminaries}
\label{sec:definitions}

We suppose that in a marketplace there are present $n$ agents,
the set of which we denote by $[n]$. Each agent $i$ has a private valuation
function $v_i$, which is a nonnegative real number for each possible winner
set $S\subset[n]$. The seller organizes a single round sealed bid auction, where
agents submit their valuations $b_i(S)$ for all possible winner sets $S$ to an auctioneer and
he then chooses agents who will obtain the good and vector of prices to
charge each of them. The auctioneer is interested in maximizing his
revenue.

For every $i\in[n]$ we impose the following quite mild requirements on $v_i$ and later in the
Section \ref{section:Extensions} we will discuss in detail why most of them are indeed necessary.

\begin{enumerate}
\item \label{i1} $v_i(S)\ge 0$.
\item \label{i2} $v_i(S)=0$ if $i\notin S.$
\item \label{i3} $v_i(S)$ is a monotone sub-additive function of $S$, i.e.
     \begin{enumerate}
     \item \label{i3_1}$v_i(S)\le v_i(R)$ if $S\subseteq R\subseteq [n]$.
     \item \label{i3_2}$v_i(S\cup R) \leq v_i(S) +v_i(R),$ for each $i\in S,R\subseteq [n]$
     \end{enumerate}
\end{enumerate}

\subsection{Mechanism Design}
Each agent in turn would like to get a positive utility as
high as possible and may lie strategically about his valuations.
The utility $u_i(S)$ of an agent $i$ for a winning set $S$ is simply the
difference of his valuation $v_i(S)$ and the price $p_i$ the auctioneer
charges $i$. Thus one of the desired properties for the auction is the
well known concept of truthfulness or incentive compatibility, i.e. the
condition that every agent maximizes his utility by truth telling.

It worth to mention here that our model is that of multi-parameter mechanism design
and, moreover, that collecting the whole bunch of values $v_i(S)$ for every
$i\in[n]$ and $S\subset [n]$ would require exponential in $n$ number of bits and thus is
inefficient. However, in the field of mechanism design there is a way to
get around such a problem of exponential input size by the broadly recognized concept
of black box value queries. The later simply means that the auctioneer instead of
getting the whole collection of bids instantly may ask instead during the mechanism execution
every agent $i$ only for a small part of his input, i.e. a number of questions about
valuation of $i$ for certain sets. We note that as usual the agent may lie
in a response on each such query. We denote the bid of $i$ by $b_i(S)$ to distinguish it
from actual valuation $v_i(S)$. Thus if we are interested in designing
computationally efficient mechanism, we can only ask in total a polynomial in $n$
number of queries.

Throughout the paper by $\calM$ we denote a mechanism
with allocation rule $\calA$ and payment rule $\calP$.
Allocation algorithm $\calA$ may ask quarries about valuations
of any agent for any possible set of winners. Thus $\calA$ has
an oracle black box access to the collection of bid functions $b_i(S)$.
For each agent $i$ in the winning set $S$ the
payment algorithm decides a price $p_i$ to charge. The utility of
agent $i$ is then $u_i=v_i(S)-p_i$ if $i\in S$ and $0$ otherwise.
To emphasize the fact that agents may report untruthfully we will use
$u_i(b_i)$ notation for the utility function in the general case and
$u_i(v_i)$ in the case of truth telling. We assume voluntary
participation for the agents, that is $u_i\ge 0$ for each $i$ who reports
the truth.

\subsection{Revenue Maximization and Possible Benchmarks}
We discuss here the problem of revenue maximization from
the seller's point of view. The revenue of the auctioneer is
simply the total payment $\sum_{i \in S} p_i$ of all buyers in the winning set.
We assume that the seller incurs no additional cost for making a copy of the
good. As a matter of fact, this assumption is essential for our
model, since unlike the classical digital auction case there is no simple
reduction of the settings with a positive price per issuing
the item to the settings with zero price.

The best revenue the seller can hope for is $\sum_{i \in
[n]} v_i ([n])$. However, it is not realistic when the seller does
not know agents' valuation functions. We follow the tradition of
the literature~\cite{FiatGHK02,Goldberg2006,FeigeFHK05,Hartline2005}
of algorithmic mechanism design on competitive auctions
with limited or unlimited supply and consider the best revenue uniform price benchmark,
which is defined as maximal revenue that auctioneer can get for a fixed uniform price for
the good. In the literature on classical competitive auctions this benchmark was called
$\calF$ and formally is defined as follows.

\begin{definition}
For the vector of agent's bids $\vecb$
\[
\calF(\vecb)=\max_{c\ge 0, S\subset[n]}\left(c\cdot |S|\Big| \forall i\in
S~~ b_i\ge c \right).
\]
\end{definition}

This definition generalizes naturally to our model with
externalities and is defined rigorously as follows.

\begin{definition}
\label{def:F} For the collection of agents' bid functions $\vecb$.
\[
\calF(\vecb)=\max_{c\ge 0, S\subset[n]}\left(c\cdot |S|\Big| \forall i\in
S~~ b_i(S)\ge c \right).
\]
\end{definition}

The important point of considering $\calF$ in the setting of classical
auctions is that the auctioneer, when is given in advance the best uniform
price, can run a truthful mechanism with corresponding revenue. It
turns out that the same mechanism works truthfully and neatly for
our model. Specifically, a seller who is given in advance the price $c$ can
begin with the set of all agents and drop one by one those agents with
negative utility ($b_i(S)- c<0$); once there are left no agents to delete the
auctioneer sells the item to all surviving buyers at the given price $c$.

Traditionally, the major question arising before auctioneer in such circumstances
is to devise a truthful mechanism which has a good approximation ratio of the
mechanism's revenue on any possible bid $\vecb$ to the revenue of the benchmark,
assuming that agents bid truthfully in the latter case. Such ratio is usually called
{\em competitive} ratio of a mechanism.
However, it was shown (cf. \cite{Goldberg2006}) that no truthful
mechanism can guarantee any constant competitive ratio w.r.t. $\calF$.
Specifically, the unbounded ratio appears on the instances where the
benchmark buys only one item at the highest price. To overcome this
obstacle, a slightly modified benchmark $\calF^{(2)}$ has been
proposed and a number of competitive mechanisms w.r.t. $\calF^{(2)}$
were obtained~\cite{FeigeFHK05,FiatGHK02,Goldberg2006,Hartline2005}. The only difference of
$\calF^{(2)}$ from $\calF$ is in one additional requirement
that at least two buyers should be in the winning set. Similarly, for
any $k\ge 2$ we may define $\calF^{(k)}$.

\begin{definition}
\[
\calF^{(k)}(\vecb)=\max_{c\ge 0, S\subset[n]}\left(c\cdot|S|\Big| |S|\ge
k,~~ \forall i\in S~~ b_i(S)\ge c \right).
\]
\end{definition}

However, in case of our model the benchmark $\calF^{(2)}$ does not imply
the existence of constant approximation truthful mechanism. In order to
illustrate that later in Section \ref{sec:char} we will introduce a couple of
new models which differ from original one by certain additional restrictions on
the domain of agent's bids. We further give a complete characterization of truthful
mechanisms for these new restricted settings substantially exploiting the fact that
every agent's bidding language is single-parameter. Later we use that characterization
to argue that no truthful mechanism can achieve constant approximation with respect to
$\calF^{(2)}$ benchmark even for these cases. On the positive side, and quite surprisingly,
we can furnish our work in the next section with the truthful mechanism which has
constant approximation ratio w.r.t. $\calF^{(3)}$ benchmark for the general case of
multi-parameter bidding.

\section{Competitive Mechanism}
\label{section:Mechanism}
Here we give a competitive truthful mechanism, that is a mechanism which
guaranties the auctioneer to get a constant fraction of the revenue, he could get for
the best fixed price benchmark
assuming that all agents bid truthfully. We call it \mainmechanism. In the mechanism
we give the good to certain agents {\em for free}, that is without any payment. The general
scheme of the mechanism is as follows.

\begin{center}
\small{}\tt{} \fbox{\parbox{6in}{\hspace{0.05in} \\[-0.05in] \mainmechanism
\begin{enumerate}
\item Put every agent at random into one of the sets $A, B, C$.
\item Denote $r_{_A}(C)$ and $r_{_B}(C)$ the largest fixed price revenues one can extract from $C$
      given that, respectfully, either $A$, or $B$ got the good for free.
\item Let $r(C)=\max\{r_{_A}(C),r_{_B}(C)\}$.
\item Sell items to agents in $A$ for free.
\item Apply \costsharing(r(C), B, A) to extract revenue $r(C)$ from set $B$
      given that $A$ got the good for free.
\end{enumerate}}}
\end{center}

Bidders in $A$ receive items for free and increase the demand of
agents from $B$. One may say that they ``advertise'' the goods and
resemble the promotion selling participants. The agents in $C$
play the role of the ``testing'' group, the only service of which is to
determine the right price. Note that we take no agents of the
testing group into the winning set, therefore, they have nothing to gain for
bidding untruthfully. The agents of $B$ appear to be the source of
the mechanism's revenue, which is being extracted from $B$ by a cost
sharing mechanism as follows.

\begin{center}
\small{}\tt{} \fbox{\parbox{6.0in}{\hspace{0.05in} \\[-0.05in] \costsharing(r,X,Y)
\begin{enumerate}
\item $S \leftarrow X$.
\item Repeat until $T=\emptyset$:
      \begin{itemize}
          \item $T \leftarrow \{i| i\in S \mbox{ and }  b_i(S \cup Y) < \frac{r}{|S|} \}$.
          \item $S \leftarrow S\setminus T$.
      \end{itemize}
\item If $S\neq\emptyset$ sell items to everyone in $S$ at $\frac{r}{|S|}$ price.
\end{enumerate}}}
\end{center}

\begin{lemma}\label{lemma:truth}
\mainmechanism $ $ is universally truthful.
\end{lemma}

\begin{proof}

The partitioning of agent set $[n]$ into $A$, $B$, $C$ does not
depend on an agent's bids. When a partition is fixed, our mechanism
becomes deterministic. Therefore, we are only left to prove
truthfulness for that deterministic part. Let us do so by passing
through the proof separately for each set $A$, $B$ and $C$.

\begin{itemize}
\item Bids of agents in $A$ do not affect the outcome of
      the mechanism. Therefore, they have no incentive to lie.

\item No agents from $C$ could gain any profit from bidding untruthfully,
      since their utilities will be zero regardless of their bids.

\item Let us note that the \costsharing\ is applied to the agents in $B$ and the
      value of $r$ does not depend on their bids, since both $r_{_A}$ and $r_{_B}$ are
      retracted from $C$ irrespectively of bids from $A$ and $B$. Also let us note that
      at each step of the cost sharing mechanism the possible payment $\frac{r}{|S|}$
      is rising, and meanwhile the valuation function, because of monotonicity
      condition, is going down. Hence, manipulating a bid does
      not help any agent to survive in the winning set and to receive a positive
      utility, if by bidding truthfully he had been dropped from it. Neither
      mis-reporting a bid could help an agent to alter the surviving set and in the same
      time remain a winner.
      The former two observations conclude the proof of truthfulness for $B$.
\end{itemize}
\end{proof}

Therefore, from now on we may assume that $b_i(S)=v_i(S).$

\begin{theorem}\label{thm:upper-bound}
\mainmechanism $ $ is universally truthful and has an expected revenue
of at least $\frac{\calF^{(3)}}{324}$.
\end{theorem}

\begin{proof}
%\noindent \emph{Proof of Theorem \ref{thm:upper-bound}.}

We are left to prove the lower bound on the competitive ratio of our
mechanism, as we have shown the truthfulness in Lemma
\ref{lemma:truth}.

For the purpose of analysis, we separate the random part of our
mechanism into two phases. In the first phase, we sieve agents
randomly into three groups $S_1$, $S_2$, $S_3$ and in the second
one, we label the groups at random by $A$, $B$ and $C$. Note that
the combination of these two phases produces exactly the same
distribution over partitions as in the mechanism.

Let $S$ be the set of winners in the optimal $\calF^{(3)}$ solution
and the best fixed price be $p^*$. For
$1\leq i\neq j \leq 3$ we may compute $r_{ij}$ the largest revenue
for a fixed price that one can extract from set $S_i$ given $S_j$ ``advertising''
the good, that is agents in $S_j$ anyway get the good for free and thus increase
the valuations of agents from $S_i$ though contribute nothing directly to the revenue.

First, let us note that the cost-sharing part of our
mechanism will extract one of these $r_{ij}$ from at least one of the
six possible labels for every sample of the sieving phase.
Indeed, let $i_0$ and $j_0$ be the indexes for which $r_{i_0 j_0}$
achieves maximum over all $r_{ij}$ and let
$k_0=\{1,2,3\}\setminus\{i_0,j_0\}$. Then the cost-sharing mechanism will
retract the revenue $r(C)=max(r_{_A}(C),r_{_B}(C))$ on the labeling
with $S_{j_0}=A$, $S_{i_0}=B$ and $S_{k_0}=C$. It turns out, as we
will prove in the following lemma, that one can get a lower bound
on this revenue within a constant factor of $r_{_\calF}(C)$; the revenue
we got from the agents of $C$ in the benchmark $\calF^{(3)}$.

\begin{lemma}\label{lemma:half}
$r(C)\geq \frac{r_{_\calF}(C)}{4}$.
\end{lemma}

\begin{proof}
Let $S_c=S\cap C$. Thus, by the definition of $\calF^{(3)}$, we have
$r_{_\calF}(C)=|S_c| \cdot p^*$ and for all $i\in S_c$, $v_i(S)\geq p^*$.

We define a subset $T$ of $S_c$ as a final result of the following
procedure. \texttt{
\begin{enumerate}
  \item $T\leftarrow \emptyset$ and $X \leftarrow \{i| i \in S_c \mbox{ and } v_i(A\cup\{i\})\geq \frac{p^*}{2} \}$.
  \item While $X \neq \emptyset$
  \begin{itemize}
    \item $T \leftarrow T \cup X$,
    \item $X \leftarrow \{i| i \in S_c \mbox{ and } v_i(A\cup T\cup\{i\} )\geq \frac{p^*}{2} \}$
  \end{itemize}
\end{enumerate}}

For any agent of $T$ we have $v_i(A\cup T )\ge\frac{p^*}{2}$
because the valuation function is monotone. Now if
$|T|\geq \frac{|S_c|}{2}$, we get the desired lower bound. Indeed,
$$r(C)\geq r_{_A}(C)\geq \frac{|S_c|}{2} \cdot \frac{p^*}{2}
=\frac{|S_c| \cdot p^*}{4}= \frac{r_{_\calF}(C)}{4}.$$

Otherwise, let $W=S_c\setminus T$. Then we have $|W|\geq
\frac{|S_c|}{2}$. For an agent $i\in W$ it holds true that $v_i(
A \cup T\cup\{i\}) < \frac{p^*}{2}$, since otherwise we should
include $i$ into $T$. However, since $i$ wins in the
optimal $\calF^{(3)}$ solution, we have $v_i(S)\geq p^*$. The
former two inequalities together with the subadditivity of $v_i(
\cdot)$ allow us to conclude that $v_i(S\setminus(A \cup T)
)\geq \frac{p^*}{2}$ for each $i\in W$. Hence, we get $v_i(B
\cup W)\geq \frac{p^*}{2}$ for each $i\in W$, since $S\setminus(A
\cup T)\subseteq B\cup W$. Therefore, we are done with the proof,
since
$$r(C)\geq r_{_B}(C)\geq |W|\cdot\frac{p^*}{2} \ge\frac{|S_c| \cdot
p^*}{4}= \frac{r_{_\calF}(C)}{4}.$$
\end{proof}

Let $k_1$, $k_2$, $k_3$ be the number of winners of the optimal
$\calF^{(3)}$ solution, respectively, in $S_1$, $S_2$, $S_3$.

For any fixed partition $S_1$, $S_2$, $S_3$ of the sieve phase by
applying Lemma \ref{lemma:half}, we get that the expected revenue of
our mechanism over a distribution of six permutations in the second
phase should be at least
$$\frac{1}{6}\cdot\frac{1}{4}\min\{k_1,k_2,k_3\} \cdot p^*.$$

In order to conclude the proof of the theorem we are only left to
estimate the expected value of $\min\{k_1,k_2,k_3\}$ from below by
some constant factor of $|S|$. The next lemma will do this for us.

\begin{lemma}\label{lemma:partition}
Let $m\ge 3$ items independently at random be put in one of
the three boxes and let $a$, $b$ and $c$ be the random variables
denoting the number of items in these boxes. Then
$\mathbb{E}[\min\{a,b,c\}]\geq \frac{2}{27}m$.
\end{lemma}

\ifabs{}{
\begin{proof}
Intuitively, it is clear that for the large $m$ the value of
$\mathbb{E}[\min\{a,b,c\}]$ should be close to $\frac{m}{3}$ (the
expectation of each random variable $a$, $b$ and $c$). More formally,
we have three random variables with dependency on them given by the
relation $a+b+c=m$. Now consider separately one of them, say
$a$. Then the distribution of $a$ is nothing else but the
distribution one may get taking the sum of independent and
identically distributed random variables $X_1,X_2,\ldots X_m$ drawn
from the Bernoulli distribution with parameters $p(1)=\frac{1}{3}$ and
$p(0)=\frac{2}{3}$.

We may use Chernoff's bounds on the probability of
$\frac{a}{m}=\frac{1}{m}\sum_{i=1}^{m}X_i$ diverging from
$p=\frac{1}{3}$ as follows.
\[
Pr\left(\frac{1}{m}\sum_{i=1}^{m}X_i\le p-\delta\right) \le \left(
\left(\frac{p}{p-\delta}\right)^{p-\delta}
\left(\frac{1-p}{1-p+\delta}\right)^{1-p+\delta} \right)^{m}
\]
Simple calculations for $p=\frac{1}{3}$ and $\delta=\frac{2}{9}$
show that for each $m\ge 17$ we will get
$$Pr\left(\frac{a}{m}\le\frac{1}{9}\right)<\frac{1}{9}.$$

Now, since the probability of the union of events is smaller than
the sum of probabilities of every event, we get
$$Pr\left(\frac{\min\{a,b,c\}}{m}\le\frac{1}{9}\right)<\frac{1}{3}.$$

Therefore,
$Pr\left(\frac{\min\{a,b,c\}}{m}\ge\frac{1}{9}\right)>\frac{2}{3}$
and
$$E\left(\frac{\min\{a,b,c\}}{m}\right)>\frac{1}{9}\cdot\frac{2}{3}=\frac{2}{27}.$$

The latter proves the lemma for $m\ge 17.$ For smaller $m$ we may
compute $\frac{\mathbb{E}[\min\{a,b,c\}]}{m}$ directly, or use more
accurate estimations on a probability and verify that
$\frac{\mathbb{E}[\min\{a,b,c\}]}{m}$ achieves its minimum when
$m=3.$
\end{proof}
}

By definition of the benchmark $F^{(3)}$ we have $m=k_1+k_2+k_3\geq
3$ and thus we can apply Lemma \ref{lemma:partition}. Combining
every bound we have so far on the expected revenue of our mechanism
we conclude the proof with the following lower bound.
\[
\frac{1}{6}\cdot\frac{1}{4}\mathbb{E}\left[\min\{k_1,k_2,k_3\}\right]\cdot
p^*\geq \frac{1}{24} \cdot  \frac{2}{27}\cdot p^* \cdot m =
\frac{F^{(3)}}{324}.
\]
\end{proof} %of theorem

\section{Restricted Single-parameter valuations}
\label{sec:char}

We introduce here a couple of special restricted cases of
the general setting with single parameter bidding language. For these models we only
specify restrictions on the valuation functions.
In each case we assume that $t_i$ is a single private parameter for agent $i$
that he submits as a bid and $w_i(S)$ and $w_{i}'(S)$ are fixed publicly known functions
for each possible winning set $S$. The models then are described as follows.

\begin{itemize}
\item {\em Additive} valuation $v_i(t_i,S)=t_i + w_i(S).$
\item {\em Scalar} valuation $v_i(t_i,S)=t_i\cdot w_i(S).$
\item {\em Linear} valuation $v_i(t_i,S)=t_i w_i(S)+ w_{i}'(S),$ i.e. combination of
      previous two.
\end{itemize}

We note that we still require that $w_i(S)=w'_i(S)=0$ if $i \not \in S$.
These settings are now single parameter domains, which is
the most well studied and understood case in mechanism design.

\subsection{A characterization}
The basic question of mechanism design is to describe
truthful mechanisms in terms of simple geometric conditions.
Given a vector of $n$ bids, $\vecb=(b_1,\ldots,b_n)$, let $b_{-i}$
denote the vector, where $b_i$ is replaced with a `?'.
%, that is,
%$$b_{-i}=(b_1,\ldots,b_{i-1},?,b_{i+1},\ldots,b_n).$$
It is well known that truthfulness implies a {\em monotonicity}
condition stating that if an agent $i$ wins for the bid vector
$\vecb=(b_{-i}, b_i)$ then she should win for any bid vector
$(b_{-i}, b'_i)$ with $b'_i\ge b_i$. In single-dimensional domains
monotonicity turns out to be a sufficient condition for
truthfulness~\cite{Archer01}, where prices are determined by the
threshold functions.

%%bid-independent allocation rule

In our model valuation of an agent may vary for different winning
sets and, thus, may depend on her bid. Nevertheless, any truthful
mechanism still has to have a bid-independent allocation rule,
although now it does not suffice for truthfulness. However, in the case
of linear valuation functions we are capable of giving a complete
characterization.
% (all proofs are deferred to appendix \ref{appendix:characterization}).
% first without any payments involved and next with
%explicit description of the payment rule.

\begin{theorem}
\label{th:allocation-characterization} In the model with linear
valuation functions $v_i(t_i,S)=t_i\cdot w_{i}(S)+w_{i}'(S)$ an
allocation rule $\calA$ may be truthfully implemented if and only if
it satisfies the following conditions:
\begin{enumerate}
\item $\calA$ is bid-independent, that is for each agent $i$,
bid vector $\vecb=(b_{-i},b_{i})$ with $i\in\calA(\vecb)$ and any
$b'_i\ge b_i$, it holds that $i \in \calA(b_{-i}, b'_i)$.

\item $\calA$ encourages asymptotically higher bids, i.e. for any fixed $b_{-i}$
and $b'_i\geq b_i$, it holds that $w_i(\calA(b_{-i}, b'_i)) \geq
w_i(\calA(b_{-i}, b_i))$.
\end{enumerate}
\end{theorem}

\ifabs{
Here we prove that these conditions are indeed necessary. The sufficiency part is deferred to the full paper, where we prove the characterization for a even more generalized family of single parameter valuation functions.

\begin{proof}
The necessity of the first monotonicity condition was known. So we prove here that the second condition is also necessary. In the truthful mechanism, an agent's payment should not depend
on her bid, if by changing it mechanism does not shift the allocated
set. We denote by $p$  the payment of agent $i$ for winner set
$\calA(b_{-i}, b_i)$ and by $p'$ the payment of agent $i$  for winner set $\calA(b_{-i}, b'_i)$.  If the agent's true value is $b_i$, by truthfulness, we
have
\[b_i\cdot w_{i}(\calA(b_{-i}, b_i))+w_{i}'(\calA(b_{-i}, b_i)) -p \geq b_i\cdot w_{i}(\calA(b_{-i}, b'_i))+w_{i}'(\calA(b_{-i}, b'_i)) -p'.\]
And if  the agent's true value is $b'_i$, we have
\[b'_i\cdot w_{i}(\calA(b_{-i}, b'_i))+w_{i}'(\calA(b_{-i}, b'_i)) -p' \geq b'_i\cdot w_{i}(\calA(b_{-i}, b_i))+w_{i}'(\calA(b_{-i}, b_i)) -p.\]
Adding these two inequalities and using the fact that $b'_i\geq b_i$, we have
%\begin{aligned}
%v_i(b_i, \calA(b_{-i}, b_i)) - v_i(b_i, \calA(b_{-i}, b'_i))\\
%%
%\geq v_i(b'_i, \calA(b_{-i}, b_i)) -v_i(b'_i, \calA(b_{-i}, b'_i)).
%\end{aligned}
%\]
\[w_i(\calA(b_{-i}, b'_i)) \geq
w_i(\calA(b_{-i}, b_i)).\]
%This contradicts the fact that $g_i(t_i,\calA(b_{-i},
%b'_i),\calA(b_{-i}, b_i))$ is strictly monotone decreasing function.
\end{proof}
}
{
\begin{proof}
We need in essence the following property, which we call
\textit{marginal monotonicity} and which holds for linear valuation
functions.

\begin{definition}
For any fixed sets $S_1$ and $S_2$, let

$g_i(t_i,S_1,S_2)=v_i(t_i,S_1)-v_i(t_i,S_2)$. Then $g_i$ as a
function of $t_i$ should be either strictly monotone (increasing or
decreasing), or constant.
\end{definition}

Thus, in fact, one can substitute in theorem
\ref{th:allocation-characterization} the requirement of valuation
function being linear for the condition of marginal monotonicity. In
the latter case the second condition of theorem
\ref{th:allocation-characterization} changes into: for any fixed
$b_{-i}$ and $b'_i\geq b_i$, it holds that $g_i(t_i,\calA(b_{-i},
b'_i),\calA(b_{-i}, b_i))$ is monotone increasing or constant. At
first we prove that this condition is indeed necessary.

\begin{proof}
If not, there has to exist $b_{-i}$ and $b'_i\geq b_i$ such that
$g_i(t_i,\calA(b_{-i}, b'_i),\calA(b_{-i}, b_i))$ is neither
monotone increasing or constant. Then by marginal monotonicity,
it is strictly monotone decreasing.

For a truthful mechanism an agent's payment should not depend
on her bid, if by changing it mechanism does not shift the allocated
set. We denote the payment of agent $i$ for winner set
$\calA(b_{-i}, b_i)$ as $p$ and for winner set $\calA(b_{-i}, b'_i)$
as $p'$.  If the agent's true value is $b_i$, by truthfulness, we
have
\[v_i(b_i, \calA(b_{-i}, b_i)) -p \geq v_i(b_i, \calA(b_{-i}, b'_i)) -p'.\]
And if  the agent's true value is $b'_i$, we have
\[v_i(b'_i, \calA(b_{-i}, b'_i)) -p' \geq v_i(b'_i, \calA(b_{-i}, b_i)) -p.\]
Adding these two inequalities, we have
\[
\begin{aligned}
v_i(b_i, \calA(b_{-i}, b_i)) - v_i(b_i, \calA(b_{-i}, b'_i))\\
\geq v_i(b'_i, \calA(b_{-i}, b_i)) -v_i(b'_i, \calA(b_{-i}, b'_i)).
\end{aligned}
\]
This contradicts the fact that $g_i(t_i,\calA(b_{-i},
b'_i),\calA(b_{-i}, b_i))$ is strictly monotone decreasing function.
\end{proof}

In the following, we prove that these two conditions are indeed sufficient
by providing an algorithm that computes payments. The payment algorithm is
determined by the allocation algorithm by the so called "Myerson
integral" \cite{Myerson1981,Archer01}. In our concrete case we can make it
more explicit. For a given bidder $i$ let us consider $S_1,S_2,\ldots,S_{N}$
as all the possible winning sets containing $i$ ($N=2^{n-1}$). We may define the order
$>_i$ on them by setting $S_k >_i S_j$ if $g_i(t_i, S_1, S_2)$ is
an increasing function in $v_i$ and $S_k <_i S_j$ if $g_i$ is
decreasing; naturally we get an equivalence relation $=_i$ if $g_i$
is constant. Therefore, one may split these $N$ sets into $m_i$
different equivalence classes, where among these different classes there is a
linear order. For convenience, we put all the sets that does not
contain $i$ into an equivalence class.

Then for each $i$ and fixed $b_{-i}$ one gets a finite partition $I_0, I_1,
\ldots, I_s$ of $[0,+\infty]$ into intervals (open, closed,
half open, half closed) and isolated points such that $[0,+\infty]=\cup_{j=0}^s
I_j$; for all $b_i$ running over $I_j$, $\calA(b_{-i}, b_i)$
could only change within the same equivalence class $\pi_j$. More specifically, there are $s+1$
equivalence classes $\pi_0,\pi_1,\ldots,\pi_s$ {\it w.r.t.} $<_i$, such that for any  $0\leq j < k \leq s$ and $S\in \pi_j$,
$S'\in \pi_k$ , we have $S<_i S'$.

Let $S_j$ be a set in $\pi_j$. We define
\[d_j=\inf_{x \in I_{j+1} }v_i(x, S_j) -\inf_{x \in I_j } v_i(x, S_j). \]
By the definition of equivalence classes, $d_j$ does not depend on
the choice of $S_j$ in $\pi_j$. Indeed, the definition of $\pi_j$
implies that $v_i(x,S)-v_i(x,S')=v_i(y,S)-v_i(y,S')$ for any
$S',S\in\pi_j$, which gives us what we need. Then the payment for a
bid $b_i \in I_{\ell}$ may be determined as follows:
\[p_i(b_i)=\inf_{x\in I_{\ell}} v_i\Big(x, \calA(b_{-i}, b_i)\Big)-\sum_{j=0}^{\ell-1}
d_j.\]
%\[p_i=\inf_{x\in I_k} v_i\Big(x, \calA(b_{-i}, b_i)\Big)-\sum_{j=0}^{k-1}
%d_j, \mbox{ where } b_i \in I_k.\]

\begin{claim}
The above payment rule makes the mechanism truthful and as a result
the conditions in Theorem \ref{th:allocation-characterization} are
also sufficient.
\end{claim}

\begin{proof}
We use $u_i(t_i, b_i)$ to denote agent i's utility when his true
value is $t_i$ and he bids $b_i$, given that $b_{-i}$  is fixed. To
prove the truthfulness it suffices to show that $u_i(t_i,t_i)\ge
u_i(t_i,b_i)$ for any $t_i$, $b_i$ and fixed $b_{-i}$. Without loss
of generality we may assume that $t_i\in I_k$ and $b_i\in I_\ell$.
For each $j$, let us pick a set $S_j$ from $\pi_j$. Then we can
write an explicit formula for $u_i(t_i,b_i)$.

\begin{eqnarray*}
u_i(t_i,b_i)&=&v_i(t_i,\calA(b_i))-p_i(b_i)\\
&=& v_i(t_i,\calA(b_i))-\inf_{x\in I_{\ell}} v_i(x, \calA(b_i))+\sum_{j=1}^{\ell-1}d_j\\
&=& \inf_{x\in I_{\ell}}\Big(v_i(t_i,\calA(b_i))- v_i(x, \calA(b_i))\Big)+\sum_{j=1}^{\ell-1}d_j\\
&=& \inf_{x\in I_{\ell}}\Big(v_i(t_i,S_{\ell})- v_i(x, S_{\ell})\Big)+\sum_{j=1}^{\ell-1}d_j\\
&=& v_i(t_i,S_{\ell})-\inf_{x\in I_{\ell}} v_i(x, S_{\ell})+\sum_{j=1}^{\ell-1}d_j.\\
\end{eqnarray*}

Similarly one can get the formula
$$u_i(t_i,t_i)=v_i(t_i,S_k)- \inf_{x\in I_k} v_i(x,S_k)+\sum_{j=1}^{k-1}d_j.$$

%$$u_i(t_i,b_i)=v_i(t_i,S_\ell)- \inf_{x\in I_\ell} v_i(x,S_\ell)+\sum_{j=1}^{\ell-1}d_j.$$

%\begin{eqnarray*}
%u_i(t_i,t_i)&=&v_i(t_i,\calA(t_i))-p_i(t_i)\\
%&=& v_i(t_i,\calA(t_i))-\inf_{x\in I_k} v_i(x, \calA(t_i))+\sum_{j=1}^{k-1}d_j\\
%&=& \inf_{x\in I_k}\Big(v_i(t_i,\calA(t_i))- v_i(x, \calA(t_i))\Big)+\sum_{j=1}^{k-1}d_j\\
%&=& \inf_{x\in I_k}\Big(v_i(t_i,S_k)- v_i(x, S_k)\Big)+\sum_{j=1}^{k-1}d_j\\
%&=& v_i(t_i,S_k)-\inf_{x\in I_k} v_i(x, S_k)+\sum_{j=1}^{k-1}d_j.\\
%\end{eqnarray*}

Before we prove $u_i(t_i,t_i)\ge u_i(t_i,b_i)$, we need the following inequality: If $S >_i
S'$ and $x>y$ then we have
\begin{equation}
\label{eq:order} v_i(x,S)-v_i(x,S')-v_i(y,S)+v_i(y,S')\ge 0
\end{equation}
This follows from the definition of $>_i$.
%\[

Let us rewrite $u_i(t_i,t_i)-u_i(t_i,b_i)$ and consider two cases.

\noindent {\bf Case 1: $\mathbf{t_i>b_i}$.} Then $k\ge\ell$ and we get that $u_i(t_i,t_i)-u_i(t_i,b_i)$ is equal to
\[v_i(t_i,S_k)-v_i(t_i,S_\ell)-\inf_{x\in I_k}v_i(x, S_k)+\inf_{x\in I_\ell}v_i(x,S_\ell)+\sum_{j=\ell}^{k-1}d_j.\]
After plugging in all formulas for $d_j=\inf\limits_{x \in I_{j+1}
}v_i(x, S_j) -\inf\limits_{x \in I_j }v_i(x, S_j)$ and rearranging
some terms we can write
\begin{eqnarray*}
%u_i(t_i,t_i)-u_i(t_i,b_i)&=&
&\ & \bigg(v_i(t_i,S_k)-v_i(t_i,
S_\ell) \\ &\ \ \ - &\inf_{x \in I_k} v_i(x,S_k)+ \inf_{x \in I_k} v_i(x,S_\ell) \bigg)\\
&+&\bigg(\inf_{x\in I_{k}} v_i(x,S_{k-1})-\inf_{x \in I_k} v_i(x,S_\ell)\\ &\ \ \ - & \inf_{x\in I_{k-1}} v_i(x,S_{k-1})+\inf_{x\in I_{k-1}} v_i(x,S_\ell)\bigg)\\
&& \quad\quad\quad\quad\quad\quad\quad\quad\quad\quad\quad + \ldots\\
&+&\bigg(\inf_{x\in I_{\ell+2}} v_i(x,S_{\ell+1})-\inf_{x\in
I_{\ell+2}} v_i(x,S_\ell)\\ &\ \ \ - & \inf_{x\in I_{\ell+1}}
v_i(x,S_{\ell+1})+\inf_{x\in I_{\ell+1}} v_i(x,S_\ell)\bigg)
\end{eqnarray*}

By applying \ref{eq:order} to each term in parentheses we get the
desired inequality.

%%u_i(t_i,t_i)=v_i(t_i,\calA(t_i))-\calP_i(t_i) &=& v_i(t_i,\calA(t_i))-v_i(x_{k-1},\calA(t_i))+\sum_{j=1}^{k-1}d_j\\
%%&=& v_i(t_i,S_k)-v_i(x_{k-1},S_k)+\sum_{j=1}^{k-1}d_j\\

\noindent {\bf Case 2: $\mathbf{t_i<b_i}$.}
 Similarly, we get that $u_i(t_i,t_i)-u_i(t_i,b_i)$ is equal to
$$v_i(t_i,S_k)-\inf_{x\in I_k} (v_i(x, S_k))-v_i(t_i,S_\ell)+\inf_{x\in I_\ell} v_i(x,S_\ell)-\sum_{j=k}^{\ell-1}d_j.$$
Rearranging terms in a different way we can write the following.
\begin{eqnarray*}
%u_i(t_i,t_i)-u_i(t_i,b_i)
&\ & \bigg(v_i(t_i,S_k)-v_i(t_i,
S_\ell)\\ &\ \ \ - & \inf_{x\in I_{k+1}}v_i(x,S_k)+\inf_{x\in I_{k+1}} v_i(x,S_\ell) \bigg)\\
&+&\bigg(\inf_{x\in I_{k+1}} v_i(x,S_{k+1})-\inf_{x\in I_{k+1}} v_i(x,S_\ell)\\ &\ \ \ - & \inf_{x\in I_{k+2}} v_i(x,S_{k+1})+\inf_{x\in I_{k+2}} v_i(x,S_\ell)\bigg)\\
&& \quad\quad\quad\quad\quad\quad\quad\quad\quad\quad\quad +\ldots\\
&+&\bigg(\inf_{x\in I_{\ell-1}} v_i(x,S_{\ell-1})-\inf_{x\in
I_{\ell-1}} v_i(x,S_\ell) \\ &\ \ \ - & \inf_{x\in I_{\ell}}
v_i(x,S_{\ell-1})+\inf_{x\in I_{\ell}} v_i(x,S_\ell)\bigg)
\end{eqnarray*}
Again inequality (\ref{eq:order}) applied to each term in brackets
concludes the proof.
\end{proof}
\end{proof}

\begin{remark}
\label{rm:issue} If all valuation functions are continuous, this is
the unique payment rule to make the mechanism truthful up to the
additive constant (as a function of $b_{-i}$) to all possible
payments of $i$. Since we assume $p_i=0$ if $i$ is not in the winner
set, then the payment is fixed as above in most cases. However, if
$i$ wins even when bidding $0$ for some fixed $b_{-i}$ then one can
reduce the payment by a fixed number in $[0, v_i(0,S_1)]$ for all
payments of $i$.
\end{remark}

\begin{remark}
Marginal monotonicity is a crucial property for our fairly simple
characterization. For example if the valuation functions are of the
form $v_i(t_i, S) = min (C, t_iw_i(S))$ one can find a truthful
mechanism with $\calA(b_{-i},x)=S_1,$ $\calA(b_{-i},y)=S_2$ and
$\calA(b_{-i},z)=S_1$ for $x<y<z$. The latter example seems to be
quite natural if an agent has a budget constraint and scalar
valuation function. That leads us to an interesting question to
characterize all truthful mechanisms in our model for broader class
of valuation functions.
\end{remark}
}

\subsection{From {\large $\calF^{(2)}$} to {\large  $\calF^{(3)}$} }
\label{subsection:F2}

Here we show that the usage of $\calF^{(2)}$ as a
benchmark may lead to an unbounded approximation ratio even for the restricted single parameter scalar valuations.
This justifies why we used a slightly modified benchmark  $\calF^{(3)}$ in Section \ref{section:Mechanism}.
%Surprisingly
%as we have shown in Section \ref{section:Mechanism} one can fix this
%problem by considering a slightly modified benchmark,
%namely $\calF^{(3)}$.

\begin{theorem}
\label{cl:example} There is no universally truthful mechanism that
can archive a constant approximation ratio w.r.t. $\calF^{(2)}$.
\end{theorem}
\begin{proof}

Consider the example of two people, such that everyone valuates the
outcome, where both have got the item, much higher than the outcome,
where only one of them getting the item, i.e.
$v_1(x,\{1\})=v_2(x,\{2\})=x$ and $v_1(x,\{12\})=v_2(x,\{12\})=M x$
for a large constant $M$. We note that these are single parameter scalar valuations.
%(One can think of two people communicating
%in life almost only with each other and who want to buy telephones;
%another example may be given by two people who want to buy
%engagement rings).

We will show that any universally truthful mechanism $\calM_{\calD}$
with a distribution $\calD$ over truthful mechanisms cannot achieve
an approximation ratio better than $M$. Each truthful mechanism
$\calM$ in $\calD$ either sells items to both bidders for some pair
of bids $(b_1,b_2)$, or for all pairs of bids sells not more than
one item. In the first case, by our characterization of truthful
mechanisms (see theorem \ref{th:allocation-characterization}),
$\calM$ should also sell two items for the bids $(x,b_2)$ and
$(b_1,y)$, where $x\ge b_1$ and $y\ge b_2$. Therefore, $\calM$ has
to sell two items for any bid $(x,y)$ with $x\ge b_1$ and $y\ge
b_2$. Let us denote respectively the first and second group of
mechanisms in $\calD$ by $\calD_1$ and $\calD_2$.

We may pick sufficiently small $\epsilon$ and consider sufficiently
large $x$, such that at least $1-\epsilon$ fraction of mechanisms in
$\calD_1$ sells two items for bids $(\frac{x}{2M},\frac{x}{2M})$.
Note that

\begin{itemize}
\item revenue of $\calF^{(2)}$ for the bids $(x,x)$ is $2Mx$,
\item revenue of $\calM$ in $\calD_2$ for the bids $(x,x)$ is not greater than $x$,
\item revenue of more than $1-\epsilon$ fraction of mechanisms in
      $\calD_1$ is not greater than $2M\frac{x}{2M}=x$.
\item revenue of the remaining $\epsilon$ fractions of mechanisms in $\calD_1$ is
      not greater than $2Mx$.
\end{itemize}
Thus we can upper bound the revenue of $\calM_{\calD}$ by
$x(1-\epsilon)+2Mx\epsilon$ while the revenue of $\calF^{(2)}$ is
$2Mx$. By choosing sufficiently large $M$ and small $\epsilon$ we
will get an arbitrary large approximation ratio.
\end{proof}

\subsection{Better Mechanism for Additive valuations}
\label{subsec:additive}

\graphicspath{{figures/}}

Assuming that each valuation function is additive, that is of the form $v_i(t_i, S)=t_i+w_i(S)$
with only one private parameter $t_i$ and publicly known additive factor $w(S_i)$.
Then the second condition in Theorem \ref{th:allocation-characterization} becomes trivial, which means that
the monotonicity condition
only by itself suffices for an allocation rule to be truthfully
implementable. 

\ifabs{}
{
\begin{corollary}
If valuation functions are additive, i.e. for each $i$ there is
exactly one equivalence class for $=_i$, the monotonicity condition
only by itself suffices for an allocation rule to be truthfully
implementable.
\end{corollary}
}

We further show that for this restricted family of valuations,
we are able to run significantly simpler mechanism with the smaller competitive ratio
comparing to $\calF^{(2)}$ instead of $\calF^{(3)}$.

\begin{theorem}
Given any $\alpha$-completive truthful mechanism $\calM_0$ for unlimited
supply auctions without externalities, one may give a
$2(1+\alpha)$-competitive truthful mechanism for markets with an
additive valuation w.r.t. benchmark $\calF^{(2)}$.
\end{theorem}

\ifabs{
We note that the theorem even does not require the subadditivity of the function $w(S_i)$.   The proof of this theorem is deferred to the full paper.
}{
\begin{proof}
We use the following mechanism
\begin{center}
\small{}\tt{} \fbox{\parbox{6.0in}{\hspace{0.05in} \\[-0.05in] $\textup{{\sc Mechanism-2}}$
\begin{enumerate}
\item At probability $\frac{1}{1+\alpha}$ give goods to everyone, charge each agent $i$ price $w_i([n])$.
\item At probability $\frac{\alpha}{1+\alpha}$, run allocation algorithm $\calA_0$
of $\calM_0$ on bid vector $\vect$; charge threshold payments according to Theorem
\ref{th:allocation-characterization}.
\end{enumerate}
}}
\end{center}
Let $\tilde{\calF}^{(2)}$ denote the benchmark's revenue for the
same vector of bids if we forget the external additive part of each
valuation. From figure \ref{f1}, we can bound the $\calF^{(2)}$ as
$\calF^{(2)}\leq 2 \tilde{\calF}^{(2)} + 2 \sum_i v_i([n])$.

Our mechanism can get the expected revenue as at least
\[
\frac{1}{1+\alpha} \sum_i v_i([n]) + \frac{\alpha}{1+\alpha}\cdot
\frac{1}{\alpha} \cdot \tilde{\calF}^{(2)}\geq \frac{1}{2(1+\alpha)}\calF^{(2)}.
\]

\begin{figure}[ht]
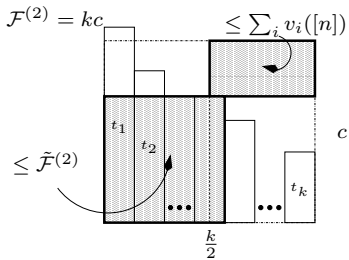

\begin{center}
\include{fig1_t}
\caption{Sort $k$ winning agents of $\calF^{(2)}$ according to their
interests. Two shaded rectangles cover at least half of
$\calF^{(2)}$ aria. Note, that we include in the rectangle
corresponding to $\tilde{\calF}^{(2)}$ at least
$\lfloor\frac{k}{2}\rfloor+1\ge 2$ agents.} \label{f1}
\end{center}
\end{figure}

\end{proof}

\noindent {\bf Remark:} For this competitive ratio, we do not need
the property that the functions $v_i(S)$ are sub-additive. We only
need the property that it is monotone.
}

\section{Discussion and Open Problems}
\label{section:Extensions}

To the best of our knowledge the model introduced in the current
paper is the first that takes into account positive externalities in
respect of studying truthful mechanism design for auctions in a worst
case revenue maximization and the first one in algorithmic community that
treats efficiently general multi-parameter case. Because of that there are
many promising ways for expansion of the model and we would like to discuss here some
possible directions. However, most of our results obtained for
such attempts are negative; thus, to get some positive results one
may try some further requirements and modifications of
the model.

\begin{enumerate}
\item \textbf{Valuations are not necessarily sub-additive.}
Then for any fixed $k$ there is no competitive mechanism with
respect to $\calF^{(k)}$. A bad instance is similar to the one in
section \ref{subsection:F2} (we let $v_i(t_i,[n])>>v_i(t_i,S)$ for each
$i$ and $S\varsubsetneq[n]$). However, one may consider relaxed
sub-additivity condition, i.e. $L(v_i(A)+v_i(B))\ge v_i(A\cup B)$
for a constant $L$ and each $i\in A,B\subset [n]$. Our mechanism will be
still working and remain competitive, though with additional factor depending on $L$.

\item \textbf{Making a copy of the good has a fixed cost for seller.}
For the original digital goods auctions one may easily make a
reduction to the setting with zero cost per copy: subtract
the cost from the agent's valuation and ignore those agents whose value
is less than zero. For our model with externalities this extension
may lead to an unbounded competitive ratio. \ifabs{Details can be founded in the full paper.}
{
\begin{claim}
If making a copy of the good has a fixed cost for seller, then
the competitive ratio may be unbounded.
\end{claim}
\begin{proof}
Let us consider the instance with $n$ bidders: $v_i(t_i, S)=|S|*t_i$ for each $i\in S\subset
[n],$ ($n>3$). We let the price for making a copy be $n$ and the
vector of true interests $\vect(\epsilon)$ be $t_i= 1+\epsilon$,
where $0<\epsilon < \frac{1}{n-1}$. Let us notice that any truthful
mechanism may extract a positive revenue on that vector only when all
bidders get into the winning set. The revenue of $\calF^{(3)}$ at
the best uniform price $n(1+\epsilon)$ will be $n^2
\epsilon=n^2(1+\epsilon)-n^2$ also with all agents being in the
winning set. This revenue is the maximum that any truthful mechanism
could get on $\vect(\epsilon)$. Now let us assume that there is a
distribution $\calD$ over truthful mechanisms and that distribution
is $L$-competitive for some constant $L$. Any truthful mechanisms
allocating the good to $[n]$ for a $\vect(\epsilon_0)$ according to
our characterization should also allocate goods to $[n]$ on any
$\vect(\epsilon)$ with $\epsilon_0<\epsilon$. Thus, for each
truthful mechanism $\calM$ we may consider the infinum of $\epsilon$
such that $\calM$ allocates goods to $[n]$ on $\vect(\epsilon)$. We
denote as $\calD_{\epsilon}$ the mechanisms in $\calD$ with such
infinum laying in $[\frac{\epsilon}{2L},\epsilon]$. Note that each
truthful mechanism in $\calD\setminus\calD_{\epsilon}$ either
allocates goods not to $[n]$ on $\vect(\epsilon)$ and, therefore,
achieves negative or zero revenue, or allocates goods to $[n]$ even
on $\vect(\frac{\epsilon}{2L})$ and, thus, on $\vect(\epsilon)$ gets
the revenue to be not more than $n^2\frac{\epsilon}{2L}$. Hence, rewriting
the condition of $L$-competitiveness on $\vect(\epsilon)$ we obtain

$$n^2\epsilon\cdot prob(\calD_{\epsilon}|\calD)+\frac{n^2\epsilon}{2L}\cdot
(1-prob(\calD_{\epsilon}|\calD))\ge\frac{1}{L}\cdot n^2\epsilon.$$

Then $prob(\calD_{\epsilon}|\calD)\ge\frac{1}{2L-1}$ for any
$\epsilon\in(0,\frac{1}{n-1})$. Taking $\epsilon$ from
$\{\frac{1}{n},\frac{1}{3Ln},\frac{1}{(3L)^2 n},\ldots\}$ we get
infinitely many disjoint sets $\calD_{\epsilon}$ with
$prob(\calD_{\epsilon}|\calD)\ge\frac{1}{2L-1}$ and arrive at a
contradiction.
\end{proof}
}

\item \textbf{Limited supply.} This direction also looks very hard to explore,
since it is not clear even how to define a benchmark. A simple algorithm
where we merely start with $[n]$ and successively remove agents with
low valuations may fail, since we could finish with a larger
number of agents than provided supply. In the latter case it is
unclear which agents we should remove next. Another difficulty with
this direction is that one may think of limited supply as of hidden
negative externality. Indeed, if an agent buys something, then besides the
increment of other's valuations she also decreases the supply, thus
probably depriving other agents of the chance to get into the winning set.

\item \textbf{Positive valuation for an agent not getting the good.} If we
      drop the condition that $v_i(S)=0$ when $i\notin S$, then
      we cannot hope for any constant competitive mechanism. For example
      one may consider simple restricted valuation function $v_i(t_i,S)=t_i \cdot |S|,~~\forall i\in [n]$
      with single private parameter $t_i$ for each agent $i$.
      Clearly, in any mechanism it will be hard to motivate any agent to pay for the good, since
      agents prefer to loose and pay nothing rather than win and pay at least something given
      that the size of winning set does not decrease. \ifabs{More detailed explanation can be found in the full paper.} {
      \begin{claim}
Let $v_i(t_i,S)=t_i \cdot |S|,~~\forall i\in [n],~S\subset [n]$. Then there is no
universally truthful competitive mechanism w.r.t. $\calF^{(k)}$ for any fixed $k$.
\end{claim}
\begin{proof}
Let's assume the contrary that there is a distribution $\calD$ of deterministic truthful mechanisms
with constant competitive ratio w.r.t. $\calF^{(k)}$.
Let $\calM=(\calA,\calP)$ be a mechanism in this distribution.
Clearly, to describe $\calA$ it suffices to specify only the size of winning set for every bid.
We may consider a bid vector $\vecb^0=(t_1^0,\dots,t_n^0)$ such that $\calA$ outputs a set $S^0$ of maximal possible size.
Note that by individual rationality the payment of each agent $i$ should not be larger than $t_i^0\cdot |S|$.
Now if agent $i$ has true type $t_i$ greater than $n\cdot t_i^0 $, while others bid $\vecb_{-i}^{0}$, then $\calA$ should
output a set $S$ of the same size as $S^0$. Indeed, by truthfulness we have
\[
t_i|S|\ge u_i(t_i,t_i)\ge
u_i(t_i, t_i^{0})= t_i |S^0|- \calP_i(\vecb^0)\ge t_i |S^0|-t_i^0\cdot n.
\]

Therefore, $t_i>n\cdot t_i^0\ge t_i(|S^0|-|S|)$ and hence $|S^0|=|S|.$
By similar argument we get that for any $\vecb=(t_1,\dots,t_n)$,
with $t_i\ge n\cdot t_i^0$, allocation rule $\calA$ should output the set of
the maximal size. As all outcomes are the same for $\vecb>n\vecb^0=(nt_1^0,\dots,nt_n^0)$,
we may also write an upper bound $n^2t_i^0$ on the payment of each agent $i$ on every such bid $\vecb$.

The revenue of $\calF^{(k)}$ on each bid $(t,\dots,t)$ is $tn^2.$
Let's take sufficiently large $t$, such that at least $(1-\epsilon)$ fraction of
mechanisms in $\calD$ output the largest possible set
on every bid vector $\vecb\ge (\frac{t}{n^2},\dots,\frac{t}{n^2})$. Then for the
bid $(t,\dots,t)$ the total payment of this $(1-\epsilon)$ fraction should be not more
than $n^2\frac{t}{n^2}=t.$ Thus the total expected revenue of $\calD$ is smaller than or equal to
$\epsilon n^2 t+(1-\epsilon)t$, while revenue of $\calF^{(k)}$ is $n^2t$. Therefore,
the competitive ratio is not more than
\[
\frac{n^2 t\epsilon+(1-\epsilon)t}{n^2t}=\epsilon+\frac{(1-\epsilon)}{n^2}.
\]
Taking $\epsilon$ sufficiently small and $n$ sufficiently large we come to a contradiction.
\end{proof}
}

\end{enumerate}

We would like to conclude the discussion with a list of open
problems

\begin{enumerate}
\item We got a constant competitive ratio w.r.t. to the fixed price benchmark.
Therefore, we think it will be an interesting research direction to obtain a competitive mechanism
with a better ratio. Also one may find it interesting to explore the
lower bounds for the new model with externalities.

\item Another important theoretical question is to give a
characterization of truthful mechanisms for general valuation
functions. In fact, the marginal monotonicity condition (see full version)
we were using for that may be not met when valuations functions are bounded
from above, e.g. budget constraint on the linear valuation.
Moreover, in such a case there exists a mechanism that cannot be
put in our characterization.

\item Truthful mechanism design for a market with externality is an
interesting and challenging research topic. In this paper, we were
studying only one particular setting. More generalization looks
interesting both for practical and theoretical points of view, for
example, negative externalities. It seems a challenging
question to find a good benchmark and design competitive
mechanisms.
\end{enumerate}

\bibliographystyle{plain}
\bibliography{game}

%\newpage
%
%\appendix
%\input{appendix}

\end{document}